%% file: Letter2014b.tex
\documentclass[journal]{IEEEtran}

\usepackage{cite,graphicx,subfigure,hyperref,amsmath,amsthm,amsfonts,mathtools,array}
\usepackage{amsthm}

\input{MatrixNotation}

\def \cconv{\overset{\scriptscriptstyle{N}}{*}}

\newcounter{TheoremCounter}
\newtheorem{theorem}[TheoremCounter]{{Theorem}}

\begin{document}

\title{FFT Interpolation from Nonuniform Samples Lying in a Regular Grid}

\author{J. Selva   
}

\maketitle

\markboth{Submitted to the IEEE Trans. on Signal Processing}{}

\begin{abstract}

This paper presents a method to interpolate a periodic band-limited signal from its samples lying at nonuniform positions in a regular grid, which is based on the FFT and has the same complexity order as this last algorithm. This kind of interpolation is usually termed ``the missing samples problem'' in the literature, and there exists a wide variety of iterative and direct methods for its solution. The one presented in this paper is a direct method that exploits the properties of the so-called erasure polynomial, and it provides a significant improvement on the most efficient method in the literature, which seems to be the burst error recovery (BER) technique of Marvasti's et al. The numerical stability and complexity of the method are evaluated numerically and compared with the pseudo-inverse and BER solutions. 

\end{abstract}

\section{Introduction}
\label{sec:i}

In a variety of applications, a band-limited signal is converted from the analog to the discrete domain, but some of the resulting samples are lost due to various causes. Then, the problem is  to interpolate the lost samples from the available ones, assuming  the average rate of the latter fulfills the Nyquist condition. Just to cite a few applications in which this problem arises, it is a task required whenever a sampled signal is sent through a packet network and there exist  losses \cite{Ghandi08}. Also, it is a basic spectral estimation problem whenever a channel spectrum must be estimated from its nonuniform samples in OFDM systems \cite{Ozdemir07,Fertl10}, (pilot-aided estimation).  It is equivalent to the error calculation step for the so-called Bose-Chaudhuri-Hocquenghem (BCH) DFT codes, in which the coding is performed in the real field, before quantization, \cite{Vaezi14,Ferreira03,Rath04}.  Finally, in time-interleaved analog-to-digital converters (TI-ADCs), some samples at arbitrary positions can be unavailable due to a jitter calibration process, and they must be recovered, \cite{Pillai14,Tsui14}.

In sampling theory, this problem is usually termed ``the missing samples problem'', and is addressed assuming the signal's bandwidth is unknown but fulfills the Nyquist condition. The basic interpolation model is then the trigonometric one, i.e, the signal is viewed as a trigonometric polynomial, and the problem reduces to computing the polynomial's coefficients and from them the missing samples \cite{Tuncer07}, \cite[Ch. 17]{Marvasti01}. As can be easily deduced, this task is equivalent to solving a linear system for which there exist various standard techniques. There are, however, two main issues. The first is the numerical stability, due to the fact that  the round-off errors accumulate heavily if there is a large number of consecutive missing samples. The second is the complexity, given that the linear system size is large and the complexity order of the standard  techniques depends cubically on it. These two drawbacks have led to the development of various direct and iterative algorithms for recovering the lost samples during the last decades. Probably, the earliest solution in the literature was the Papoulis-Gerchberg algorithm \cite{Papoulis75,Ferreira94}, which is an iterative method based on the FFT. Standard techniques like the conjugate gradient and Lagrange interpolation methods have also been employed \cite[Sec. 3]{Marvasti00}. \cite{Ghandi08}  presents another iterative method and several ways to speed up its convergence using extrapolation. The BER technique of Marvasti et al. in \cite{Marvasti00} seems to be the most efficient technique to date. This technique is numerically stable and achieves the complexity order $\FO(NP)$, where $N$ is the total number of samples and $P$ the number of known ones. This complexity order is a clear improvement relative to the complexity order of the standard methods which is $\FO(P^3)$.

The purpose of this paper is to present a new direct solution for the missing sampling problem whose complexity order is $\FO(N\log N)$. If $a$ denotes the ratio of total to known samples $N/P$ and is assumed constant, then the complexity of the BER technique is $\FO(N^2)$ while that of the proposed method is $\FO(N\log N)$. Thus, the proposed method provides a significant improvement in terms of complexity. Actually, its arithmetic operations count is up to  factor twenty smaller than the corresponding count of the BER technique, for typical values of $N$. The method proposed in this paper is based on two theorems. The first gives a procedure to obtain the missing samples which consists of two FFTs plus three weighting operations. The coefficients of two of the three weighting operations depend on the sampling positions, and thus the procedure is efficient but only usable if these last coefficients have been pre-computed. The second theorem provides a solution to this last shortcoming, by specifying a procedure to compute the weighting coefficients in just two FFTs plus the computation of one complex exponential per output sample. The combination of these two theorems yields the proposed method whose complexity is $\FO(N\log N)$.  

The paper has been organized as follows. In the next sub-section, we introduce the notation and recollect several basic results about periodic signals and the FFT, that will be instrumental in the paper.  Then, in Sec. \ref{sec:tm} we introduce the missing samples problem and the BER technique. Afterward, we present  in Secs.  \ref{sec:pm} and \ref{sec:ce} the two theorems that make up the proposed method. The complexity order of the BER and proposed method are then discussed in Sec. \ref{sec:ca}. Finally, Sec. \ref{sec:ne} compares the standard pseudo-inverse, BER, and proposed methods numerically in terms of numerical stability and computational burden.

\subsection{Notation and basic concepts}
\label{sec:n}

We will employ the following notation:

\begin{itemize}
\item Throughout the paper, $t\in \mathbb R$ will denote the time variable, and $n$, $p$ and $q$ integer variables.
\item Definitions of new symbols and functions will be written using '$\equiv$'.
\item Vectors will be denoted in bold face, ($\vs$, $\vd$).
\item $[\vv]_n$ will denote the $n$th component of vector $\vv$.
\item For integer ${M\geq 0}$, $I_M$ will denote the set 
\[
\{0,\,1,\ldots,\, M-1\}.
\]
\item For a sequence $S_p$ and an index set $A$, the notation $\{S_p: p\in A\}$ will specify the set of values $S_p$ for indices $p$ in $A$.  

\item $\text{DFT}\{\vv\}$ and $\text{IDFT}\{\vv\}$ will respectively denote the DFT and IDFT of  vector $\vv$, but computed using a fast algorithm based on the FFT. Note that there exist fast algorithms of this kind for any vector length, like the Chirp transform, \cite[Sec. 6.9.2]{Oppenheim89}.

\item For two sets, $A$ and $B$, $A\backslash B$ will denote the set of elements in $A$ not in $B$. 

\item $\va\odot\vb$ will denote the component-wise product of $\va$ and $\vb$, i.e, for $\va$ and $\vb$ of equal length, $[\va\odot\vb]_n=[\va]_n[\vb]_n$. 

\item For two $N$-period discrete sequences, $\Fa(n)$ and $\Fb(n)$, $(\Fa \cconv \Fb)(n)$ will denote their cyclic convolution, defined by
\be{eq:336}\nonumber 
(\Fa\cconv\Fb)(n)\equiv \sum_{p=q}^{q+N-1} \Fa(p)\Fb(n-p),
\ee
where $q$ can be any integer, given that $\Fa(n)$ and $\Fb(n)$ have period $N$. The cyclic convolution can be efficiently evaluated using the FFT by means of the formula
\be{eq:337}
(\Fa\cconv\Fb)(n)=[\text{IDFT}\{\text{DFT}\{\va\}\odot\text{DFT}
\{\vb\}\}]_{n+1},\;n\in I_N,
\ee
where
\be{eq:338}\nonumber
[\va]_{n+1}\equiv \Fa(n),\;\;[\vb]_{n+1}\equiv \Fb(n),\; n\in I_N.
\ee
\end{itemize}

In the paper, we will exploit several basic results about the evaluation of periodic band-limited signals using the FFT. For integer $N\geq 1$, $\mathcal F_N$ will denote the set of signals whose Fourier series is restricted to the index range $[0,N-1]$; specifically, $\mathcal F_N$ contains the signals of the form
\be{eq:269} 
\Fv(t)=\sum_{p=0}^{N-1}V_p\Fe^{j2\pi pt/N},
\ee
where $V_p\in \mathbb C$ and $t\in \mathbb R$. 

For $\Fv(t)\in \mathcal F_N$, consider the following vectors of samples of $\Fv(t)$ and $\Fv'(t)$, and Fourier coefficients $V_p$, 
\be{eq:270} 
[\vv]_{n+1}\equiv \Fv(n),\; [\vv']_{n+1}\equiv \Fv'(n),\;\;[\tilde \vv]_{p+1}\equiv V_p,\; (n,\,p\in I_N).
\ee
As is well known,  we can switch from $\vv$ to $\tilde \vv$ and vice versa through the equations  
\begin{gather}
\label{eq:376}\nonumber\vv=N \text{IDFT}\{\tilde \vv\}, \\
\label{eq:366}\nonumber \tilde \vv =\frac{1}{N}\text{DFT}\{\vv\}.
\end{gather}
We may express this relation using set notation in the following way,
\be{eq:404}
\{v_n:\, n\in I_N\} \xrightarrow{\text{DFT}} \{NV_p: p\in I_N\},
\ee
where we interpret these sets as ordered. 

$\mathcal F_N$ is  closed under differentiation, i.e, if ${\Fv(t)\in \mathcal F_N}$ then ${\Fv'(t)\in \mathcal F_N}$. This property is obvious since the Fourier coefficients of $\Fv'(t)$ are $\{j2\pi p V_p/N:\, p\in I_N\}$. A consequence of this property is that we may compute the values  $\Fv'(n)$, $n\in I_N$, using the DFT/IDFT pair. More precisely, we have  
\be{eq:278}
\vv'=\text{IDFT}\{\text{DFT}\{\vv\}\odot\vd\},
\ee
where 
\be{eq:335}
[\vd]_{p+1}\equiv j2\pi p/N,\;  p \in I_N.
\ee
\section{The missing samples problem}
\label{sec:tm}

A basic interpolator for a band-limited signal $\Fs_o(t)$ is the trigonometric one, i.e, it consists of viewing $\Fs_o(t)$ as a trigonometric polynomial of the form 
\be{eq:383}
\Fs_o(t)\approx\sum_{p=p_1}^{p_1+P-1} S_{o,p}\Fe^{j2\pi p t/T},
\ee
where we assume that $\Fs_o(t)$ is interpolated in the range $[0,T]$ with $T>0$,  $S_{o,p}$ denotes the $p$th coefficient, $p_1$ the first polynomial index, and $P>0$ the number of coefficients. If (\ref{eq:383}) is sufficiently accurate and $\Fs_o(t)$ is sampled with period $T/N$ for an integer $N\geq P$, then it is well-known that the coefficients $S_{o,p}$ and the value of $\Fs_o(t)$ at any $t\in [0,T]$ can be efficiently computed from the set of samples $\{\Fs_o(nT/N),n\in I_N\}$ using algorithms from the FFT family \cite{Frigo05}. 

In some applications, however, $N-P$ samples from the set $\{\Fs_o(nT/N),\,n\in I_N\}$ are lost due to various causes, and then the problem consists of recovering these missing samples in a numerically stable way and with low computational burden from the known ones. More precisely, if $J$ denotes the indices $n$ of the known samples, then $J$ has $P$ elements and the objective is to obtain the samples $\{\Fs_o(nT/N),\, n\in J^c\}$, where $J^c$ is the complement of $J$ relative to $I_N$,
\be{eq:384}\nonumber 
J^c\equiv I_N\backslash J.
\ee

In this problem, the initial index $p_1$ and the time period $T$ are irrelevant, given that we may scale $\Fs_o(t)$ so that its period is $N$ and its first frequency is zero. So, in order to simplify the notation, we may state the problem in terms of the following normalized signal
\be{eq:291}\nonumber
\Fs(t)\equiv \Fs_o(tT/N)\Fe^{-j2\pi p_1t/N}.
\ee
From (\ref{eq:383}) we have that $\Fs(t)$ has the form 
\be{eq:292}
\Fs(t)=\sum_{p=0}^{P-1} S_p \Fe^{j2\pi p t/N},
\ee
where $S_p\equiv S_{o,p+p_1}$.  In terms of $\Fs(t)$, the problem consists of computing the samples $\{\Fs(n):\, n\in J^c\}$,  assuming the samples $\{\Fs(n):\, n\in J\}$ are known. 

As can be readily checked, the solution to this problem just involves the inversion of the linear system
\be{eq:380}
\Fs(n)=\sum_{p=0}^{P-1} S_p \Fe^{j2\pi p t/N},\;\; n\in J, 
\ee
in which the unknowns are the coefficients $S_p$, followed by the computation of the desired samples using (\ref{eq:292}) for $t\in J^c$. The inversion  can  in principle be tackled using conventional linear algebra techniques whose complexity order is $\FO(P^3)$, \cite[Ch. 3]{Golub96}. It must be noted that (\ref{eq:380}) is often ill conditioned, specially if there exist long sequences of missing samples, and it is then necessary to resort to the pseudo-inverse. The high $\FO(P^3)$ complexity of conventional methods has led to the development of a variety of iterative and non-iterative methods with  lower complexity during the last decades; (see \cite[Ch. 17]{Marvasti00} for a review on this topic).

In \cite{Marvasti00}, Marvasti et al. presented the so-called BER technique for this problem whose complexity order is $\FO(NP)$. This order is a clear improvement relative to the $\FO(P^3)$ order of the standard solutions, and relative to other methods like the Lagrange interpolation and conjugate gradient methods. The key of the BER method consists of two relations between the following three polynomials:

\begin{itemize}
\item $\Fs_J(t)$: Element of $\mathcal F_N$ such that $\Fs_J(n)=\Fs(n)$ if $n\in J$ and $\Fs_J(n)=0$ if $n\in J^c$.
\item $\Fs_{J^c}(t)$: Polynomial with the same definition as $\Fs_J(t)$ but with $J$ and $J^c$ switched. 
\item $\Fphi(t)$: Erasure polynomial. This is the monic element of $\mathcal F_N$ of degree $N-P$ that has one simple zero at each of the  instants of the missing samples (set $J^c$), i.e, the polynomial 
\be{eq:260}  
\Fphi(t)\equiv \prod_{n\in J^c}(\Fe^{j2\pi t/N}-\Fe^{j2\pi n/N}).
\ee
\end{itemize}

To introduce the first relation in the BER technique, note  that to compute the desired samples $\{\Fs(n): n\in J^c\}$ is equivalent to compute ${\{\Fs_{J^c}(n): n\in I_N\}}$, given that ${\Fs_{J^c}(n)=0}$  if ${n\in J}$. Additionally, from the definitions of ${\Fs_J(t)}$ and ${\Fs_{J^c}(t)}$, it is clear that 
\be{eq:385}
\Fs(n)=\Fs_J(n)+\Fs_{J^c}(n), \, n\in I_N.
\ee
This equation can be written in the coefficients (frequency) domain using (\ref{eq:404}),
\be{eq:405}
S_p=S_{J,p}+S_{J^c,p},\, p\in I_N,
\ee
where $S_{J,p}$ and $S_{J^c,p}$ respectively denote the Fourier coefficients of $\Fs_J(t)$ and $\Fs_{J^c}(t)$. But $S_p=0$ if $P\leq p<N$ and, therefore, (\ref{eq:405}) implies 
\be{eq:386}
S_{J^c,p}=-S_{J,p}, P\leq p<N. 
\ee
So, the DFT of the samples $\{\Fs_J(n): n\in I_N\}$ gives partial information about $\Fs_{J^c}(t)$, namely its coefficients $S_{J^c,p}$ for $P\leq p<N$.

The second relation links $\Fs_{J^c}(t)$ with the erasure polynomial $\Fphi(t)$ and is the following
\be{eq:387}
\Fs_{J^c}(n)\Fphi(n)=0, \, n\in I_N.
\ee
This relation is also a direct consequence of the definitions of $\Fs_{J^c}(t)$ and $\Fphi(t)$, given that $J\cup J^c=I_N$, $\Fs_{J^c}(n)=0$ if $n\in J$, and $\Fphi(n)=0$ if $n\in J^c$. If we take (\ref{eq:387}) to the frequency domain using the DFT (\ref{eq:404}), we have that (\ref{eq:387}) is turned into a cyclic convolution of the coefficients of $\Fs_{J^c}(t)$ and $\Fphi(t)$. More precisely, we have
\be{eq:388}
\sum_{p=0}^{N-P}\phi_{N-P-p}S_{J^c,q+p}=0,\;\; q\in \mathbb Z,
\ee
where $\Fphi_p$ denotes the coefficients of $\Fphi(t)$, and we take $S_{J^c,p}$ as a periodic sequence, i.e, $S_{J^c,p+N}=S_{J^c,p}$, $p\in \mathbb Z$. This second relation can be written as a recursive formula for computing $S_{J^c,q}$, if $S_{J^c,q+p}$ is known for $1\leq p <N-P$. For this, just note from (\ref{eq:260}) that $\Fphi_{N-P}=1$ and solve for $S_{J^c,q}$ in (\ref{eq:388}),
\be{eq:389}
S_{J^c,q}=-\sum_{p=1}^{N-P}\phi_{N-P-p}S_{J^c,q+p}.
\ee
We have that (\ref{eq:386}) already provides the coefficients $S_{J^c,q+p}$ in this sum if $q=P$. So we may recursively apply this last formula for $q=P, P-1,\ldots, 0$, in order to compute the missing coefficients $S_{J^c,q}$, $0\leq q <P$.

Finally, from $\{S_{J^c,p}: p\in I_N\}$ we obtain the desired samples $\{\Fs_{J^c}(n):\, n\in J^c\}$ through one inverse DFT,
\be{eq:390}
\Fs(n)=\Fs_{J^c}(n)=[\text{IDFT}\{\vs_{J^c}\}]_{n+1},\;\; n\in J^c,
\ee
where
\be{eq:392}\nonumber
[\tilde\vs_{J^c}]_{p+1}\equiv S_{J^c,p}, \; p\in I_N.
\ee

We can see in this method that the operation of inserting zeros, either in a vector or using the erasure polynomial, is the key to obtaining an efficient solution. Actually, the zero insertion in the definitions of $\Fs_J(t)$ and $\Fs_{J^c}(t)$ permits the use of the DFT in going from (\ref{eq:385}) to (\ref{eq:405}). And the multiplication by the erasure polynomial in (\ref{eq:387}) produces a zero sequence and the corresponding cyclic convolution in (\ref{eq:388}). There is, however, a more powerful way to exploit this zero insertion property, that leads to a method entirely based on the DFT and weighting operations with complexity $\FO(N)$. The method is based on considering the properties of the signal $\Fs(t)\Fphi(t)$. This method is presented in the next section and yields the desired samples in just two DFTs, if some samples of $\Fphi(t)$ and its derivative are known.

\section{Proposed method for fixed sampling positions}
\label{sec:pm}

We have the following theorem.

\begin{theorem}
\label{th:2}

The desired samples are given by the formula
\be{eq:328}
\Fs(n)=\frac{1}{\Fphi'(n)}[\text{\textnormal{IDFT}}\{\text{\textnormal{DFT}}
\{\vs_\phi\}\odot\vd\}]_{n+1},
\;\; n\in J^c,
\ee
where $\vd$ was defined in (\ref{eq:335}) and 
\be{eq:329}
[\vs_\phi]_{n+1}\equiv
\left\{\begin{array}{lll}
\Fs(n)\phi(n) & \textrm{if} &  n\in J \\
0 & \textrm{if} & n\in J^c.
\end{array}
\right. 
\ee
\end{theorem}

This theorem specifies the method proposed in this paper to compute the desired samples $\{\Fs(n)$:\, $n\in J^c\}$,  if the values of $\Fphi(n)$ and $1/\Fphi'(n)$ appearing in (\ref{eq:328}) and (\ref{eq:329}) are available. For implementing it, it is only necessary to form the nonuniformly zero-padded vector $\Fs_\phi$ in (\ref{eq:329}), and then perform the steps specified in (\ref{eq:328}), i.e,

\begin{enumerate}
\item Compute the DFT of $\vs_\phi$.
\item Weight the result component-wise using $\vd$.
\item Compute the inverse DFT.
\item Multiply the samples with $n\in J^c$ by $1/\Fphi'(n)$.
\end{enumerate}

If what is required is the set of Fourier coefficients ${\{S_p:\,p\in I_P\}}$, then they can be computed from ${\{\Fs(n):\,n\in I_N\}}$, through one FFT using the formula
\be{eq:340}\nonumber
S_p=\text{DFT}\{\vs_1\},\; p\in I_P, 
\ee
with
\be{eq:341}\nonumber
[\vs_1]_{n+1}\equiv \Fs((N/Q)n),\; n\in I_Q,
\ee
where $Q$ is the smallest divisor of $N$ such that $Q\geq P$. 

\begin{proof}[Proof of theorem \ref{th:2}]
Consider the signal 
\be{eq:244}\nonumber
\Fs_\phi(t)\equiv \Fs(t)\Fphi(t)
\ee
and two key facts related with it. The first is that we know its value at all instants in the regular grid $I_N$. This is so because either $n\in J$ and then both factors of the product $\Fs(n)\phi(n)$ are known, or $n\in J^c$ and then $\Fs(n)\Fphi(n)=0$ because $\Fphi(n)=0$. As a consequence, we have enough information to form the vector  $\vs_\phi$ in (\ref{eq:329}), akin to $\vv$ in (\ref{eq:270}), given that the only samples of $\Fs(t)$ appearing in (\ref{eq:329}) are the known ones, [$\Fs(n)$, $n\in J$].  

The second fact is that $\Fs_\phi(t)$ belongs to $\mathcal F_N$. We can see that this is
so if we view (\ref{eq:292}) and (\ref{eq:260}) as polynomials in the variable
$z=\Fe^{j2\pi t/N}$. Since the right-hand side of (\ref{eq:292}) has degree $P-1$ and
(\ref{eq:260}) has degree $N-P$ (number of elements of $J^c$), then $\Fs(t)\Fphi(t)$ has
degree $N-1$ in $z$. In other words, $\Fs_\phi(t)$ has the form in (\ref{eq:269}). As a
consequence, we may compute the derivative samples of $\Fs_\phi(t)$ using
(\ref{eq:278}). We have
\be{eq:280}
\vs_\phi'=\text{IDFT}\{\text{DFT}\{\vs_\phi\}\odot\vd\},
\ee
where 
\be{eq:281}\nonumber
[\vs_\phi']_{n+1}\equiv \Fs_\phi'(n), \;\;n\in I_N.
\ee

Finally, the product differentiation rule allows us to obtain the desired samples
$\Fs(n)$, $n\in J^c$, from $\vs_\phi'$, given that $\Fphi(t)$ has placed zeros at the
desired instants $n\in J^c$. Specifically, since $\Fphi(n)=0$ if $n\in J^c$, we have
\be{eq:284}\nonumber
\Fs_\phi'(n) =\Fs'(n)\Fphi(n)+\Fs(n)\Fphi'(n)
=\Fs(n)\Fphi'(n).
\ee
So, solving for $\Fs(n)$ we obtain
\be{eq:249}\nonumber
\Fs(n)=\frac{\Fs_\phi'(n)}{\Fphi'(n)}=\frac{[\vs_\phi']_{n+1}}{\Fphi'(n)},\;\; n\in J^c.
\ee
Note that the division by $\Fphi'(n)$ is valid because the instants $n\in J^c$ are simple zeros of $\Fphi(n)$. The theorem's formula in (\ref{eq:328}) is the result of substituting (\ref{eq:280}) into this last equation.
\end{proof}

\section{Computation of the erasure polynomial weights $\Fphi(n)$ and $\Fphi'(n)$}
\label{sec:ce}

In Theorem \ref{th:2}, the samples of $\Fphi(n)$ and $\Fphi'(n)$ depend on the sampling scheme and, therefore, they must be re-computed whenever the set $J$ changes. If this re-computation is performed using (\ref{eq:260}) directly, then the cost of obtaining ${\{\Fphi(n):\,n\in J\}}$ is $\FO((N-P)P)$. As to the samples ${\{\Fphi'(n):\, n\in J^c\}}$, they can be computed from the derivative of (\ref{eq:260}),
\be{eq:332}\nonumber 
\Fphi'(t)=\frac{j2\pi}{N}\Fe^{j2\pi t/N}\sum_{k\in J^c}\prod_{n\in J^c\setminus \{k\}}(\Fe^{j2\pi t/N}-\Fe^{j2\pi n/N})
\ee
with complexity $\FO((N-P)^2(N-P-1))$. These complexities are too high for real-time systems. The following theorem presents a method to compute these values with complexity $\FO(N\log N)$. It involves the computation of two size-$N$ FFTs and $N$ complex exponentials.

\begin{theorem}
\label{th:1}

Consider the $N$-period sequence specified by 
\be{eq:319} 
\Falp(n)\equiv
\left\{
\begin{array}{ll}
\log(1-\Fe^{-j2\pi n/N}),& 1\leq n<N\\
0,& n=0
\end{array}
\right.
\ee
and $\Falp(n+N)=\Falp(n)$, $n\in \mathbb{Z}$. Let $\Fbet(n)$ denote the cyclic convolution
\be{eq:320}
\Fbet(n)\equiv (\mathrm 1_{J^c}\cconv\Falp)(n),
\ee
where $\mathrm 1_{J^c}(n)$ is the cyclic indicator sequence for $J^c$, defined by
\be{eq:321}
\mathrm{1}_{J^c}(n)\equiv
\left\{
\begin{array}{lll}
1 &\textrm{if}& n\in J^c\\
0 &\textrm{if}& n\in J
\end{array}
\right.
\ee
and $\mathrm{1}_{J^c}(n)=\mathrm{1}_{J^c}(n+N)$, $n\in\mathbb{Z}$. The samples of $\Fphi(t)$ and $\Fphi'(t)$ required in theorem \ref{th:2} are given by 
\begin{gather}
\label{eq:325}
\Fphi(n)=\exp\Big(-\frac{j2\pi nP}{N}+\Fbet(n) \Big),\; n\in J,\\
\label{eq:326}
\Fphi'(n)=\frac{j2\pi}{N}\exp\Big(-\frac{j2\pi nP}{N}+\Fbet(n) \Big),\; n\in J^c.
\end{gather}

\end{theorem}

Note that the sequence $\Falp(n)$ is independent of the sampling scheme and, therefore, it can be computed offline. This theorem implies that  the computation of the required samples of $\Fphi(n)$ and $\Fphi'(n)$ just requires the cyclic convolution in (\ref{eq:320}) and the computation of one complex exponential per sample. Since the cyclic convolution can be performed using the FFT [Eq. (\ref{eq:337})], the total computational cost is $\FO(N\log N)$. In computing the cyclic convolution, the DFT of the sequence $\Falp(n)$ can be spared, given that it can be performed offline. So to update $\phi(n)$ and $\Fphi'(n)$ just requires two FFTs. 

\begin{proof}[Proof of theorem \ref{th:1}.]
Let us write (\ref{eq:260}) in terms of $\Falp(n)$, taking into account that $J^c$ has $N-P$ elements. If $n\in J$, we have:
\begin{align}
\label{eq:334}
\Fphi(n)&= \prod_{k\in J^c}(\Fe^{j2\pi n/N}-\Fe^{j2\pi k/N}) \\
&= \prod_{k\in J^c}(\Fe^{j2\pi n/N}(1-\Fe^{j2\pi (k-n)/N})) \nonumber\\
&= \prod_{k\in J^c}(\Fe^{j2\pi n/N}\Fe^{\Falp (n-k)}) \nonumber\\
&= \Fe^{j2\pi n (N-P)/N}\prod_{k\in J^c}\Fe^{\Falp (n-k)} \nonumber \\
\label{eq:294}
&= \exp\Big(-\frac{j2\pi nP}{N}+\sum_{k\in J^c}\Falp(n-k)\Big).
\end{align}
In this last step, note that the summation is the cyclic convolution of $\Falp(n)$ with the indicator sequence of $J^c$ in (\ref{eq:321}); i.e,
\be{eq:322}\nonumber
\sum_{k\in J^c}\Falp(n-k)=(\mathrm 1_{J^c}\cconv \Falp)(n),
\ee
and we have from (\ref{eq:294})
\be{eq:327}
\Fphi(n)=\exp\Big(-\frac{j2\pi nP}{N}+ (\mathrm 1_{J^c}\cconv \Falp)(n)\Big),\; n\in J.
\ee
Thus we have proved (\ref{eq:325}).

For deriving (\ref{eq:326}), we must consider first the signal $\Fphi_1(t)$ with the same definition as $\Fphi(t)$ in (\ref{eq:260}), but with $J$ in place of $J^c$, i.e, the signal
\be{eq:295} 
\Fphi_1(t)\equiv \prod_{n\in J}(\Fe^{j2\pi t/N}-\Fe^{j2\pi n/N}).
\ee
For $\Fphi_1(t)$, we may repeat the derivations in (\ref{eq:334}) to (\ref{eq:327}) already performed for $\Fphi(t)$ and, as can be easily checked, the result is the formula in (\ref{eq:327}) but with $J$ and $J^c$ switched and $N-P$ in place of $P$ in the first term of the exponent. Specifically, we obtain
\be{eq:312}
\Fphi_1(n)=\exp\Big(-\frac{j2\pi n(N-P)}{N}+(1_{J}
\overset{N}{*}\Falp)(n)\Big),\; n\in J^c,
\ee
where $1_J(n)$ is the indicator sequence of $J$, defined by $1_J(n+N)=1_J(n)$, $n\in \mathbb Z$, and
\be{eq:313}\nonumber
\mathrm{1}_{J}(n)\equiv
\left\{
\begin{array}{lll}
1 &\textrm{if}& n\in J\\
0 &\textrm{if}& n \in J^c.
\end{array}
\right.
\ee

Next, we require two simple results about $\Falp(n)$ and the indicators $\mathrm 1_{J^c}(n)$ and $\mathrm 1_{J}(n)$. The first is the property
\be{eq:307}
\sum_{n=0}^{N-1}
\Falp(n)=\log(N), 
\ee
which is proved in Ap. \ref{ap:1}. The second is the fact that we may write (\ref{eq:312}) in terms of $1_{J^c}(n)$ instead of $1_J(n)$, because these two indicator functions are complementary; i.e, since $J\cup J^c=I_N$ and $J\cap J^c=\emptyset$, we have
\be{eq:315}
1_{\text{all}}(n)=\mathrm{1}_J(n)+\mathrm{1}_{J^c}(n),
\ee
where $1_{\text{all}}(n)$ is the all-ones sequence. 

Now, using (\ref{eq:307}) and (\ref{eq:315}) we have that $(\mathrm{1}_J\cconv\Falp)(n)$ can be obtained from $(\mathrm{1}_{J^c}\cconv\Falp)(n)$:
\begin{align}
(\mathrm{1}_J\cconv\Falp)(n)&=(1_{\text{all}}\cconv\Falp)(n)-(\mathrm{1}_{J^c}
\cconv\Falp)(n) \nonumber \\
&=\log(N)-(\mathrm{1}_{J^c}\cconv\Falp)(n). \nonumber
\end{align}
And substituting this formula into (\ref{eq:312}), we obtain a result of the form in (\ref{eq:327}) but for $\Fphi_1(n)$, 
\be{eq:317}
\Fphi_1(n)=N\exp\Big(-\frac{j2\pi n(N-P)}{N}-(\mathrm{1}_{J^c}\cconv\Falp)(n)\Big),\; 
n\in J^c.
\ee

Let us derive the formula for $\Fphi'(n)$, $n\in J^c$. For this, consider the product $\phi(t)\Fphi_1(t)$. From (\ref{eq:334}) and (\ref{eq:295}), we have that this product is a monic polynomial whose  root set is $\{ \Fe^{j2\pi n/N},\, n\in I_N\}$, given that $J\cup J^c=I_N$ and $J\cap J^c=\emptyset$. So, we have
\begin{align}
\label{eq:296}
\Fphi(t)\Fphi_1(t)&=\prod_{n=0}^{N-1} (\Fe^{j2\pi t/N}-\Fe^{j2\pi n/N}) \\
&=\left.\prod_{n=0}^{N-1}(z-\Fe^{j2\pi n/N})\right|_{z=\Fe^{j2\pi t/N}} \nonumber \\
&=(z^N-1)|_{z=\Fe^{j2\pi t/N}}
=\Fe^{j2\pi t}-1. \nonumber
\end{align}
In this derivation, we have used the fact that the monic $N$th-order polynomial with root set $\{ \Fe^{j2\pi n/N},\, n\in I_N\}$ is $z^N-1$. 

Next, let us apply the product differentiation rule to the equation derived in (\ref{eq:296}),
\be{eq:298}\nonumber
\Fphi(t)\Fphi_1(t)=\Fe^{j2\pi t}-1,
\ee
at $t=n$, $n\in J^c$. For its left-hand side, we have
\begin{align}
(\Fphi(t)\Fphi_1(t))'_{t=n}&=\big(\Fphi'(t)\Fphi_1(t)+\Fphi(t)\Fphi'_1(t)\big)_{t=n}
\nonumber \\
\label{eq:302}
&=\Fphi'(n)\Fphi_1(n),
\end{align}
given that $\Fphi(n)=0$ if $n\in J^c$. And for its right-hand side, we have
\be{eq:299}
(\Fe^{j2\pi t}-1)'_{t=n}=j2\pi.
\ee
So, the combination of (\ref{eq:302}) and (\ref{eq:299}) yields
\be{eq:300}\nonumber
\Fphi'(n)=\frac{j2\pi}{\Fphi_1 (n)},\;n\in J^c.
\ee
Finally, substituting (\ref{eq:317}) into this last formula we obtain
\be{eq:318}\nonumber
\Fphi'(n)=\frac{j2\pi}{N}\exp\Big(-\frac{j2\pi nP}{N}+ (\mathrm{1}_{Jc} \cconv \Falp)(n)\Big),\; n\in J^c,
\ee
which is (\ref{eq:326}).
\end{proof}

\section{Complexity analysis}
\label{sec:ca}

In this section, we present counts of the number of floating point operations (flops) for both the proposed method and the BER technique. Since the complexity of  basic operations like multiplication and complex exponential may vary wildly with the hardware implementation, we have employed the convention in Fig. \ref{fig:3} for measuring the complexity. 
\begin{figure}\centerline{
\begin{tabular}{l|c}
Operation & Flops\\\hline
Real sum & 1\\
Complex sum & 2\\
Real multiplication & 1\\
Complex multiplication & 6\\
Complex exponential & 7\\
Size-$N$ FFT, IFFT & $5N\log_2N$\\
\end{tabular}}
\caption{\label{fig:3} Flop counts for basic operations.}
\end{figure}

The flop count of each step in the BER technique, as explained in Sec. \ref{sec:tm}, is the following:

\begin{itemize}
\item Computation of ${\{\Fphi(n):\,n\in J\}}$ using (\ref{eq:260}), 
\be{eq:399}\nonumber
10P(N-P)-11P+3.
\ee
\item DFT of the sequence $\{\phi(n):\, n\in I_N\}$ for obtaining the coefficients ${\{\phi_p:\, p\in I_{N-P}\}}$ in (\ref{eq:389}): $5 N \log_2N$.

\item DFT of sequence ${\{\Fs_J(n): \, n\in I_N\}}$, for computing $\{S_{J,p}:\, p\in I_N\}$: $5 N \log_2N$.

\item Computation of recursive formula in (\ref{eq:389}),
\be{eq:401}\nonumber
8 P(N-P)-P.
\ee
\item Inverse DFT for obtaining the final result $\{\Fs_{J^c}(n):\,  n\in I_N\}$: $5 N \log_2N$.

\end{itemize}

The total cost of the BER technique is the following
\be{eq:403}
18P(N-P)-12P+3+15 N\log_2N.
\ee
The implementation of the proposed method has the following flop counts:

\begin{itemize}
\item Computation of $\{\Falp(n):\, n\in I_N\}$ in (\ref{eq:319}). We assume zero cost for this operation, given that it can be performed offline.

\item Computation of $\{\Fbet(n):\, n\in I_N\}$ in (\ref{eq:320}). This operation involves two FFTs plus $N$ complex multiplications. The cost is
\be{eq:393}\nonumber 
10 N\log_2N + 8N-1.
\ee

\item Computation of samples $\phi(n)$, $n\in J$, in (\ref{eq:329}) using (\ref{eq:325}). We assume the factor $-j2\pi P/N$ in the exponent of (\ref{eq:325}) has been pre-computed. The cost of this operation is 
\be{eq:394}\nonumber
18 P-3.
\ee
\item Computation of second factor in (\ref{eq:328}). This operation involves two DFTs and $N$ real-to-complex products with total cost
\be{eq:395}\nonumber
10 N \log_2N+4N-2.
\ee
Computation of $\{1/\phi'(n):\, n\in J^c\}$ from (\ref{eq:326}), and product with the output of the previous step. The cost is 
\be{eq:396}\nonumber  
20(N-P)-4.
\ee
\end{itemize}

The total cost of the proposed technique is the following 
\be{eq:397}\nonumber
20 \log_2N + 32N -2P-10.
\ee
By comparing (\ref{eq:403}) with this last equation, we can readily see the complexity of the proposed method is free of quadratic terms, while the complexity of the BER techniques is dominated by these terms when $P$ is separated from $0$ and $N$. 

\section{Numerical examples}
\label{sec:ne}

\subsection{Numerical stability}
\label{sec:ns}

The linear system in (\ref{eq:380}) is ill conditioned if there are long sequences of missing samples. This facts makes conventional inversion methods like Gaussian elimination unstable numerically for (\ref{eq:380}). So, in order to validate the method proposed in this paper, it is necessary to evaluate the accumulation of round-off errors. For this, we compare the following three methods in the sequel using double precision arithmetic:

\begin{itemize}
\item \emph{Pseudo-inverse method:} Based on solving (\ref{eq:380}) for unknowns $S_p$ using the pseudo-inverse, and then computing ${\{\Fs(n): \,n\in J^c\}}$, using (\ref{eq:292}). 

\item \emph{BER technique:} Combination of (\ref{eq:386}), (\ref{eq:389}), and (\ref{eq:390}).
\item \emph{Proposed method:} Method in theorem \ref{th:2} using the computation procedure for $\Fphi(n)$ and $\Fphi'(n)$ in theorem \ref{th:1}. 

\end{itemize}

In the examples that follow, we employ test signals of the form in (\ref{eq:292}) with
$S_p=S_{R,p}+jS_{I,p}$, where $S_{R,p}$ and $S_{I,p}$ are independent realizations of a
uniform distribution in the interval $[-1,1]$. The figures are based on 100 Monte Carlo
trials.

We present two examples. In the first, we assume a sampling grid which is the result of
shifting the samples of a uniform grid (jittered sampling). In the second, we address the
extrapolation problem, i.e, the sampling grid has a long gap that must be filled.

\subsubsection{Round-off error for a jittered sampling scheme}

In this example, we fix an oversampling factor $a=8$ and relate $N$ and $P$ through
$N=aP$. Then, we take sampling instants $t_p=8 p+u_p$, $p=0,\, 1,\ldots, P-1$ where $u_p$
is randomly taken from the set $\{0,1,\ldots \,, a-1\}$ with uniform distribution
(jittered sampling). 
\begin{figure}
\centerline{\includegraphics{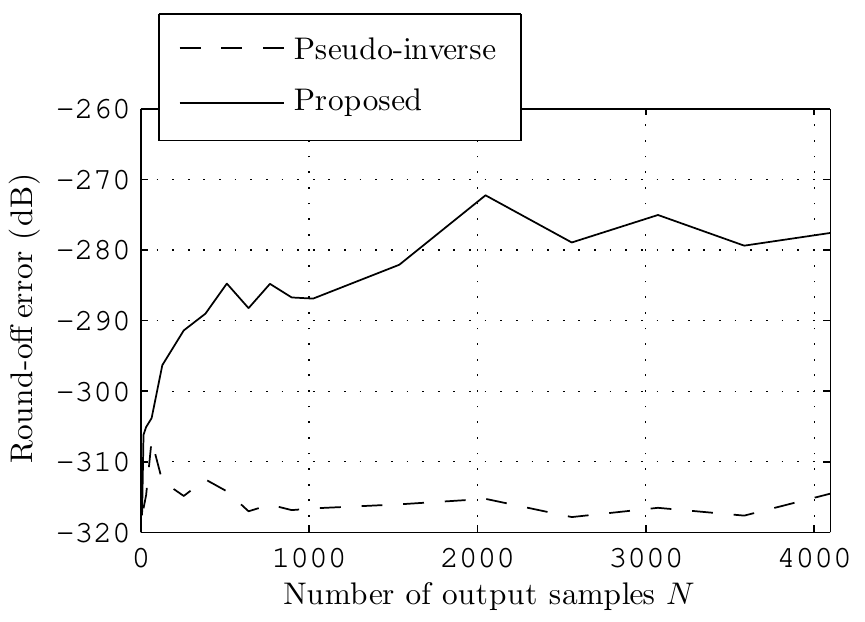}}
\caption{\label{fig:1} Round-off error versus number of output samples ($N$) for the
  proposed and pseudo-inverse methods.}
\end{figure}
Fig. \ref{fig:1} shows the round-off error versus the number of output samples $N$. The ordinate in this figure is the largest error among the $N-P$ interpolated samples. The proposed method improves on the BER technique slightly, and these two methods show a slight accuracy loss (one to two decimal digits) relative to the pseudo-inverse solution.  The error of the proposed method is sufficiently small for most applications.

\subsubsection{Round-off error for extrapolation}

In this example, we fix $N=64$ and take as input samples those with instants $0,\,
1,\ldots,\, P-1$. The objective is to interpolate the signal at instants $P,\, P+1,\ldots,
N-1$. Fig. \ref{fig:2}
\begin{figure}
\centerline{\includegraphics{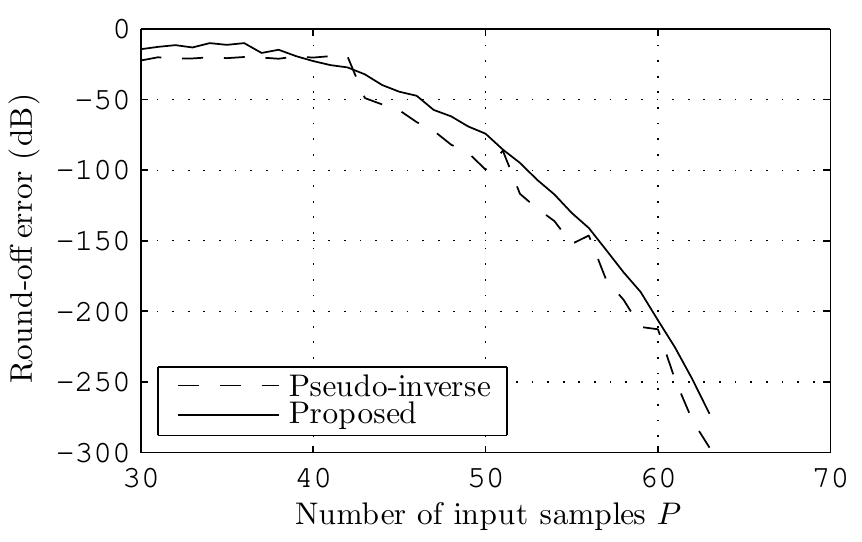}}
\caption{\label{fig:2} Maximum round-off error versus number of input samples $P$.}
\end{figure}
shows the maximum round-off error versus $P$. Note that there is little difference between the performances of the three methods, with the BER technique having a slightly better performance. 

\subsection{Computational burden}
\label{sec:cb}

In this section, we evaluate the computational burden of the proposed method relative to the BER technique, using the results in Sec. \ref{sec:ca}.

\subsubsection{Complexity versus grid size $N$}
\label{sec:cvg}

Fig. \ref{fig:3} shows the flop counts for the proposed and BER methods versus the grid size, assuming $a=N/P=8$. There are two variants of the proposed method in this figure. In variant ``Prop. A'', $\Fbet(n)$ in (\ref{eq:320}) is computed using the FFT (\ref{eq:337}), while in variant ``Prop. B'' (\ref{eq:320}) is evaluated directly. The proposed methods shows a clear improvement relative to the BER technique. For large $N$, ``Prop. A'' is roughly a factor 16 less complex than the BER technique. Also, note that ``Prop. B'' improves on ``Prop. A'' for small $N$. This is due to the fact that the cyclic convolution in (\ref{eq:320}) can be directly evaluated without any multiplications. 

\begin{figure}
\centerline{\includegraphics{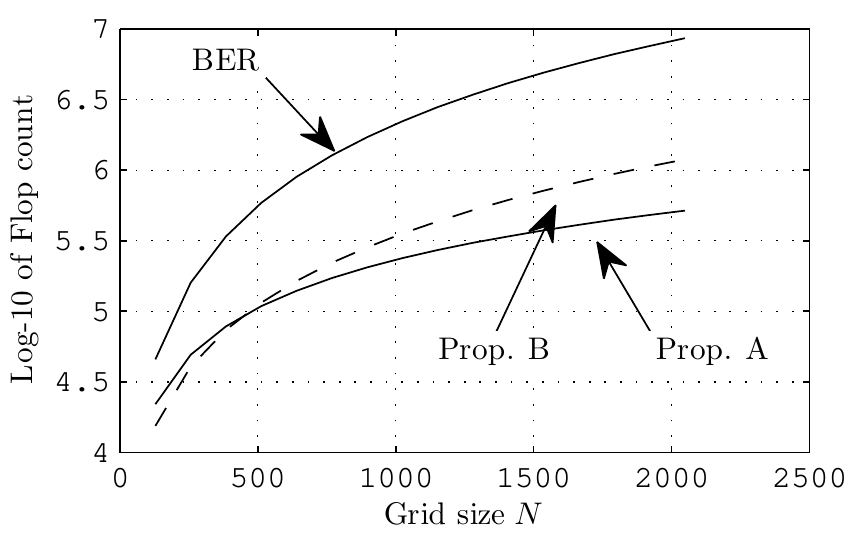}}
\caption{\label{fig:3} Complexity versus grid size $N$ for two variants of the proposed method and the BER technique. Variant ``Prop. A'' computes $\Fbet(n)$ in (\ref{eq:320}) using the FFT, while variant ``Prop. B'' performs this last computation by directly evaluating the convolution in (\ref{eq:320}).}
\end{figure}

\subsubsection{Complexity versus $N/P$ ratio}
\label{sec:cvn}

Fig. \ref{fig:5} shows the ratio 
\be{eq:406}\nonumber
\frac{\text{BER tech. flop count}}{\text{``Prop. A'' flop count}},
\ee
versus the factor $a=N/P$ for $N=1024$, where ``Prop A'' was described in the previous sub-section. This figure shows that ``Prop. A'' improves on the BER technique for all $a$ values except for the very small or very large. Actually, the BER technique is more efficient only if $a<0.0049$ or $a>0.96$, ($P\leq 5$ or $P\geq 1019$). The maximum improvement is factor 20 roughly.
\begin{figure}
\centerline{\includegraphics{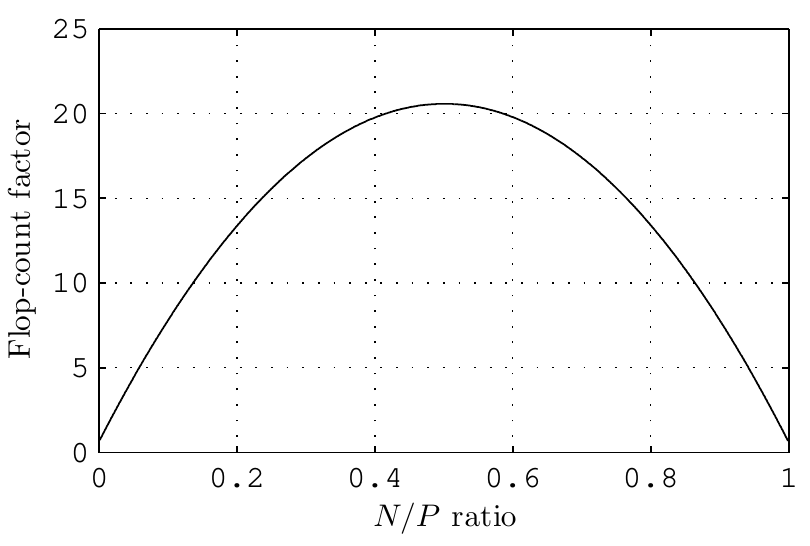}}
\caption{\label{fig:5} Improvement factor of the proposed method relative to the BER technique in terms of flop count, (proposed method's count / BER technique's count).}
\end{figure}

\subsubsection{Complexity compared with the zero-padding FFT algorithm }
\label{sec:ccz}

If $N/P$ is an integer and $J$ is a regular grid with spacing $N/P$, then the missing samples problem can be solved using the zero-padding FFT (ZP-FFT) algorithm, \cite[Sec. 3.11]{Lyons01}. Fig. \ref{fig:4} shows the complexity of this well-known algorithm and that of the method in this paper. The curve ``Proposed, no weight comp.'' is the count of ``Prop. A'' but discounting the complexity of computing $\{\Fphi(n):\, n\in J\}$ and $\{\Fphi'(n):\, n\in J^c\}$, given that the sampling grid is constant. This figure shows that the proposed method is, in rough terms, only factor two more complex than ZP-FFT if the weight factors are available, and factor 4 if not. 
\begin{figure}
\centerline{\includegraphics{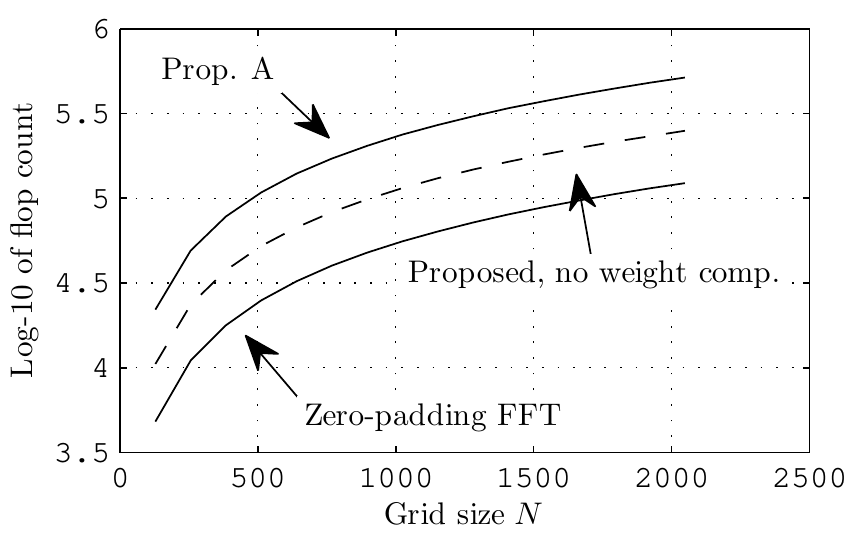}}
\caption{\label{fig:4} Complexity versus the zero-padding FFT algorithm.}
\end{figure}

\section{Conclusions}

We have presented a solution for the missing samples problem based on the FFT. The method has complexity $\FO(N\log N)$ and consists of four FFTs plus several operations of order $\FO(N)$. It provides a significant improvement on the burst error recovery (BER) technique, which seems to be the most efficient method in the literature. For typical values of $N$, the complexity is reduced up to factor 18, relative to this last technique. The method has been assessed in terms of numerical stability and computational burden numerically.

\appendices

\section{Proof of (\ref{eq:307})}
\label{ap:1}

In order to prove (\ref{eq:307}), write the summation as the logarithm of a polynomial in $z=\Fe^{j2\pi t/N}$:
\begin{align}
\label{eq:308}
\sum_{n=0}^{N-1}\Falp(n)=\log \Big(\prod_{n=1}^{N-1}(z-\Fe^{-j2\pi n/N})\Big) \Big|_{z=1}.
\end{align}
Note that $z^N-1$ is the monic polynomial with roots $\Fe^{j2\pi n/N}$, $n=0, 1,\ldots, N-1$, and these roots also appear in (\ref{eq:308}), except for the root at $z=1$. So we have that the polynomial in (\ref{eq:308}) is actually $(z^N-1)/(z-1)$ and
\begin{align}
\sum_{n=0}^{N-1}\Falp(n)&=\log \Big(\frac{z^N-1}{z-1}\Big)\Big|_{z=1} \nonumber\\
&=\log \Big(\sum_{n=0}^{N-1}z^n\Big)\Big|_{z=1}=\log(N).\nonumber
\end{align}

\bibliographystyle{IEEEbib}

\bibliography{../../../Utilities/LaTeX/Bibliography}
\end{document}

%% file: MatrixNotation.tex
\newlength{\Lpr}
\newsavebox{\Bpr}

\newcommand{\V}[1]{\mbox{\boldmath$\mathbf{#1}$\unboldmath}}

\newcommand{\bdm}{\begin{displaymath}}
\newcommand{\edm}{\end{displaymath}}

\newcommand{\be}[1]{\begin{equation} \label{#1}}
\newcommand{\ee}{\end{equation}}

\newcommand{\bae}[3]{
\begin{equation} \label{#1}
\renewcommand{\arraystretch}{#2}
\begin{array}{#3}}

\newcommand{\eae}{\end{array}\end{equation}}

\newcommand{\baen}[2]{
\begin{displaymath} 
\renewcommand{\arraystretch}{#1}
\begin{array}{#2}}

\newcommand{\eaen}{\end{array}\end{displaymath}}

\newcommand{\DefLetter}[4]{
\newcommand{#1}{\ensuremath{\V{#2}}} 
\newcommand{#3}{\ensuremath{\V{#4}}} 
}

\DefLetter{\vzer}{0}{\mzer}{0}
\DefLetter{\vone}{1}{\mone}{1}
\DefLetter{\va}{a}{\ma}{A}
\DefLetter{\vb}{b}{\mb}{B}
\DefLetter{\vc}{c}{\mc}{C}
\DefLetter{\vd}{d}{\md}{D}
\DefLetter{\ve}{e}{\me}{E}
\DefLetter{\vf}{f}{\mf}{F}
\DefLetter{\vg}{g}{\mg}{G}
\DefLetter{\vh}{h}{\mh}{H}
\DefLetter{\vi}{i}{\mi}{I}
\DefLetter{\vj}{j}{\mj}{J}
\DefLetter{\vk}{k}{\mk}{K}
\DefLetter{\vl}{l}{\ml}{L}
\DefLetter{\vm}{m}{\mm}{M}
\DefLetter{\vn}{n}{\mn}{N}
\DefLetter{\vpr}{p}{\mpr}{P}
\DefLetter{\vq}{q}{\mq}{Q}
\DefLetter{\vr}{r}{\mr}{R}
\DefLetter{\vs}{s}{\ms}{S}
\DefLetter{\vt}{t}{\mt}{T}
\DefLetter{\vur}{u}{\mur}{U}
\DefLetter{\vv}{v}{\mv}{V}
\DefLetter{\vw}{w}{\mw}{W}
\DefLetter{\vx}{x}{\mx}{X}
\DefLetter{\vy}{y}{\my}{Y}
\DefLetter{\vz}{z}{\mz}{Z}

\DefLetter{\vdel}{\delta}{\mdel}{\Delta}
\DefLetter{\vphi}{\phi}{\mphi}{\Phi}
\DefLetter{\vpsi}{\psi}{\mpsi}{\Psi}
\DefLetter{\vrho}{\rho}{\mrho}{\Lambda}
\DefLetter{\vxi}{\xi}{\mxi}{\Xi}

\DefLetter{\valpha}{\alpha}{\malpha}{\Alpha}
\DefLetter{\vbeta}{\beta}{\mbeta}{\Beta}
\DefLetter{\vlam}{\lambda}{\mlam}{\Lambda}
\DefLetter{\vsig}{\sigma}{\msig}{\Sigma}
\DefLetter{\vtau}{\tau}{\mtau}{\tau}
\DefLetter{\vtheta}{\theta}{\mtheta}{\Theta}
\DefLetter{\vome}{\omega}{\mome}{\Omega}
\DefLetter{\vzero}{0}{\mzero}{0}
\DefLetter{\vgam}{\gamma}{\mgam}{\Gamma}
\DefLetter{\veps}{\epsilon}{\meps}{\Epsilon}
\DefLetter{\veta}{\eta}{\meta}{\Eta}

\newcommand{\DefFuncLetter}[2]{
\newcommand{#1}{\ensuremath{{\mathrm{#2}}}} 
}

\DefFuncLetter{\Fzer}{0}
\DefFuncLetter{\Fa}{a}
\DefFuncLetter{\FA}{A}
\DefFuncLetter{\Fb}{b}
\DefFuncLetter{\Fc}{c}
\DefFuncLetter{\FC}{C}
\DefFuncLetter{\Fd}{d}
\DefFuncLetter{\Fe}{e}
\DefFuncLetter{\Ff}{f}
\DefFuncLetter{\Fg}{g}
\DefFuncLetter{\FG}{G}
\DefFuncLetter{\Fh}{h}
\DefFuncLetter{\FH}{H}
\DefFuncLetter{\Fi}{i}
\DefFuncLetter{\Fk}{k}
\DefFuncLetter{\Fl}{l}
\DefFuncLetter{\FL}{L}
\DefFuncLetter{\Fm}{m}
\DefFuncLetter{\Fn}{n}
\DefFuncLetter{\Fnr}{n}
\DefFuncLetter{\FN}{N}
\DefFuncLetter{\Fo}{o}
\DefFuncLetter{\FO}{O}
\DefFuncLetter{\Fpr}{p}
\DefFuncLetter{\FPr}{P}
\DefFuncLetter{\Fq}{q}
\DefFuncLetter{\Fr}{r}
\DefFuncLetter{\Fs}{s}
\DefFuncLetter{\FS}{S}
\DefFuncLetter{\Ft}{t}
\DefFuncLetter{\FT}{T}
\DefFuncLetter{\Fu}{u}
\DefFuncLetter{\FU}{U}
\DefFuncLetter{\Fv}{v}
\DefFuncLetter{\Fw}{w}
\DefFuncLetter{\FW}{W}
\DefFuncLetter{\Fx}{x}
\DefFuncLetter{\Fy}{y}
\DefFuncLetter{\FY}{Y}
\DefFuncLetter{\Fz}{z}
\DefFuncLetter{\FZ}{Z}
\DefFuncLetter{\Falp}{\alpha}
\DefFuncLetter{\Fbet}{\beta}
\DefFuncLetter{\Fchi}{\chi}
\DefFuncLetter{\Fdel}{\delta}
\DefFuncLetter{\Fzet}{\zeta}
\DefFuncLetter{\FEps}{\Epsilon}
\DefFuncLetter{\Feta}{\eta}
\DefFuncLetter{\Fphi}{\phi}
\DefFuncLetter{\FPhi}{\Phi}
\DefFuncLetter{\Fpsi}{\psi}
\DefFuncLetter{\FPsi}{\Psi}
\DefFuncLetter{\Fgam}{\gamma}
\DefFuncLetter{\FGam}{\Gamma}
\DefFuncLetter{\Flam}{\lambda}
\DefFuncLetter{\FLam}{\Lambda}
\DefFuncLetter{\Fsig}{\sigma}
\DefFuncLetter{\Ftau}{\tau}
\DefFuncLetter{\Fome}{\omega}
\DefFuncLetter{\Feps}{\epsilon}
\DefFuncLetter{\Fthe}{\theta}
\DefFuncLetter{\Fvar}{\vartheta}

\DefFuncLetter{\FB}{B}
\DefFuncLetter{\FD}{D}
\DefFuncLetter{\FE}{E}
\DefFuncLetter{\FF}{F}
\DefFuncLetter{\FI}{I}
\DefFuncLetter{\FJ}{J}
\DefFuncLetter{\FM}{M}
\DefFuncLetter{\FR}{R}
\DefFuncLetter{\FV}{V}
\DefFuncLetter{\FX}{X}

\newcommand{\DefCalLetter}[2]{
\newcommand{#1}{\ensuremath{\mathcal{#2}}} 
}
\DefCalLetter{\CC}{C}
\DefCalLetter{\CD}{D}

\DefCalLetter{\CS}{S}
\DefCalLetter{\CV}{V}

\newcommand{\DefSubLetter}[2]{
\newcommand{#1}{\mathrm{#2}} 
}

\DefSubLetter{\slzer}{0}
\DefSubLetter{\sla}{a}
\DefSubLetter{\slA}{A}
\DefSubLetter{\slb}{b}
\DefSubLetter{\slB}{B}
\DefSubLetter{\slc}{c}
\DefSubLetter{\slC}{C}
\DefSubLetter{\sld}{d}
\DefSubLetter{\slD}{D}
\DefSubLetter{\sle}{e}
\DefSubLetter{\slE}{E}
\DefSubLetter{\slf}{f}
\DefSubLetter{\slF}{F}
\DefSubLetter{\slg}{g}
\DefSubLetter{\slG}{G}
\DefSubLetter{\slh}{h}
\DefSubLetter{\slH}{H}
\DefSubLetter{\sli}{i}
\DefSubLetter{\slI}{I}
\DefSubLetter{\slk}{k}
\DefSubLetter{\sll}{l}
\DefSubLetter{\slL}{L}
\DefSubLetter{\slm}{m}
\DefSubLetter{\slM}{M}
\DefSubLetter{\sln}{n}
\DefSubLetter{\slnr}{n}
\DefSubLetter{\slN}{N}
\DefSubLetter{\slo}{o}
\DefSubLetter{\slp}{p}
\DefSubLetter{\slP}{P}
\DefSubLetter{\slq}{q}
\DefSubLetter{\slQ}{Q}
\DefSubLetter{\slr}{r}
\DefSubLetter{\slR}{R}
\DefSubLetter{\sls}{s}
\DefSubLetter{\slS}{S}
\DefSubLetter{\slt}{t}
\DefSubLetter{\slT}{T}
\DefSubLetter{\slu}{u}
\DefSubLetter{\slU}{U}
\DefSubLetter{\slv}{v}
\DefSubLetter{\slw}{w}
\DefSubLetter{\slW}{W}
\DefSubLetter{\slx}{x}
\DefSubLetter{\slX}{X}
\DefSubLetter{\sly}{y}
\DefSubLetter{\slY}{Y}
\DefSubLetter{\slz}{z}
\DefSubLetter{\slZ}{Z}

\DefSubLetter{\slalp}{\alpha}
\DefSubLetter{\slbet}{\beta}
\DefSubLetter{\sldel}{\delta}
\DefSubLetter{\slDel}{\Delta}
\DefSubLetter{\sleps}{\epsilon}
\DefSubLetter{\slgam}{\gamma}
\DefSubLetter{\slphi}{\phi}
\DefSubLetter{\sltau}{\tau}
\DefSubLetter{\slxi}{\xi}
\DefSubLetter{\slthe}{\theta}
